%% file: techrep10-predPktCtrl.tex
\documentclass[a4paper]{article}
\pdfoutput=1
\input{common_env}

\usepackage[ireport]{KTHEEtitlepage}

\graphicspath{{figures/}}
\begin{document}

\ititle{Network Estimation and Packet Delivery Prediction for Control
  over Wireless Mesh Networks}
\isubtitle{The work was supported by the EU project FeedNetBack,
the Swedish Research Council, the Swedish Strategic Research
Foundation, the Swedish Governmental Agency for Innovation Systems,
and the Knut and Alice Wallenberg Foundation.}
\iauthor{Phoebus Chen, Chithrupa Ramesh, and Karl H. Johansson}
\idate{2010}
\irefnr{TRITA-EE:043}

\iaddress{ACCESS Linnaeus Centre\\
  Automatic Control\\
  School of Electrical Engineering\\
  KTH Royal Institute of Technology\\
  SE-100 44 Stockholm, Sweden}
\makeititle

\title{Networked Controller Design using Packet Delivery Prediction in Mesh Networks}

\input{sections/abstract.tex}

\input{sections/introduction}
\input{sections/problem_formulation}
\input{sections/network_estimator}

\input{sections/controller_synthesis}

\input{sections/examples_simulations}
\input{sections/network_model_selection}

\input{sections/conclusions}

\bibliography{techrep10-predPktCtrl}
\bibliographystyle{IEEEtran}

\appendix
\onecolumn

\end{document}

%% file: common_env.tex
\usepackage[cmex10]{amsmath} 
\usepackage{amssymb} 
\usepackage{amsthm}
\usepackage{bm}      

\usepackage[mathscr]{eucal}

\usepackage{url}
\usepackage{ifpdf}
\ifpdf
\usepackage[pdftex]{graphicx,color} 
\usepackage{float} 
\else
\usepackage{epsfig} 
\fi

\usepackage{algorithm} 
\usepackage{algpseudocode}

\usepackage{fixltx2e} 

\def\Ce{{\mathscr C}}

\def\Ne{{\mathscr N}}
\def\Pe{{\mathscr P}}
\def\Ve{{\mathscr V}}

\def\R{{\mathbb R}}        
\def\Sbb{{\mathbb S}}      

\def\Ec{{\mathcal E}}
\def\Fc{{\mathcal F}}

\def\Ic{{\mathcal I}}

\def\Vc{{\mathcal V}}

\def\IID{{\textrm{IID}}}
\def\ON{{\textrm{ON}}}
\def\sopt{{\textrm{sopt}}}
\def\comp{{\textrm{comp}}}

\newcommand{\vect}[1]{\bm{#1}} 

\DeclareMathOperator*{\E}{{\mathbb E}}        
\let\Pr\undefined 
\DeclareMathOperator{\Pr}{{\mathbb P}}        

\DeclareMathOperator{\tr}{trace}

\newcommand{\D}{\displaystyle}

\newcommand{\Sst}{\scriptstyle} 

\def\registered{{\ooalign{\hfil\raise .00ex\hbox{\scriptsize R}\hfil\crcr\mathhexbox20D}}}

\newtheorem{theorem}{Theorem}[section]

\newtheorem{property}{Property}

\newenvironment{remark}[1][Remark]{\begin{trivlist}
\item[\hskip \labelsep {\bfseries #1}]}{\end{trivlist}}



%% file: sections/abstract.tex
\begin{abstract}                %

Much of the current theory of networked control systems uses
simple point-to-point communication models as an abstraction of the
underlying network.  As a result, the controller has very limited
information on the network conditions and performs suboptimally.  This
work models the underlying wireless multihop mesh network as a graph of links
with transmission success probabilities, and uses a recursive Bayesian
estimator to provide packet delivery predictions to the controller.
The predictions are a joint probability distribution on future packet
delivery sequences, and thus capture correlations between successive
packet deliveries.  We look at finite horizon LQG control over a lossy
actuation channel and a perfect sensing channel, both without delay, to
study how the controller can compensate for predicted network outages.

\end{abstract}

%% file: sections/introduction.tex
\section{Introduction}
\label{sec:introduction}

Increasingly, control systems are operated over large-scale, networked
infrastructures.  In fact, several companies today are introducing
devices that communicate over low-power wireless mesh networks for
industrial automation and process control
\cite{WINA:url,ISA-SP100:url}.  While wireless mesh networks can
connect control processes that are physically spread out over a large
space to save wiring costs, these networks are difficult to design,
provision, and manage \cite{chlamtac:2003,bruno:2005}.
Furthermore, wireless communication is inherently unreliable,
introducing packet losses and delays, which are detrimental to control
system performance and stability.

Research in the area of Networked Control Systems (NCSs)
\cite{hespanha-procIEEE:2007} addresses how to design control systems
which can account for the lossy, delayed communication channels
introduced by a network.  Traditional tasks in control systems design,
like stability/performance analysis and controller/estimator
synthesis, are revisited, with network models providing statistics
about packet losses and delays.  In the process, the studies highlight
the benefits and drawbacks of different system architectures.  For
example, Figure~\ref{fig:dist_ctrl_mesh_net} depicts the general system
architecture of a networked control system over a mesh network
proposed by Robinson and Kumar \cite{robinson:2007}.  A fundamental architecture problem
is how to choose the best location to place the controllers, if they
can be placed at any of the sensors, actuators, or communication relay
nodes in the network.  One insight from Schenato et al. \cite{schenato-procIEEE:2007}
is that if the controller can know whether the control packet reaches
the actuator, e.g., we place the controller at the
actuator, then the optimal LQG controller and estimator can be
designed separately (the separation principle).

\begin{figure}[hb]
\begin{center}
\includegraphics[width=10cm]{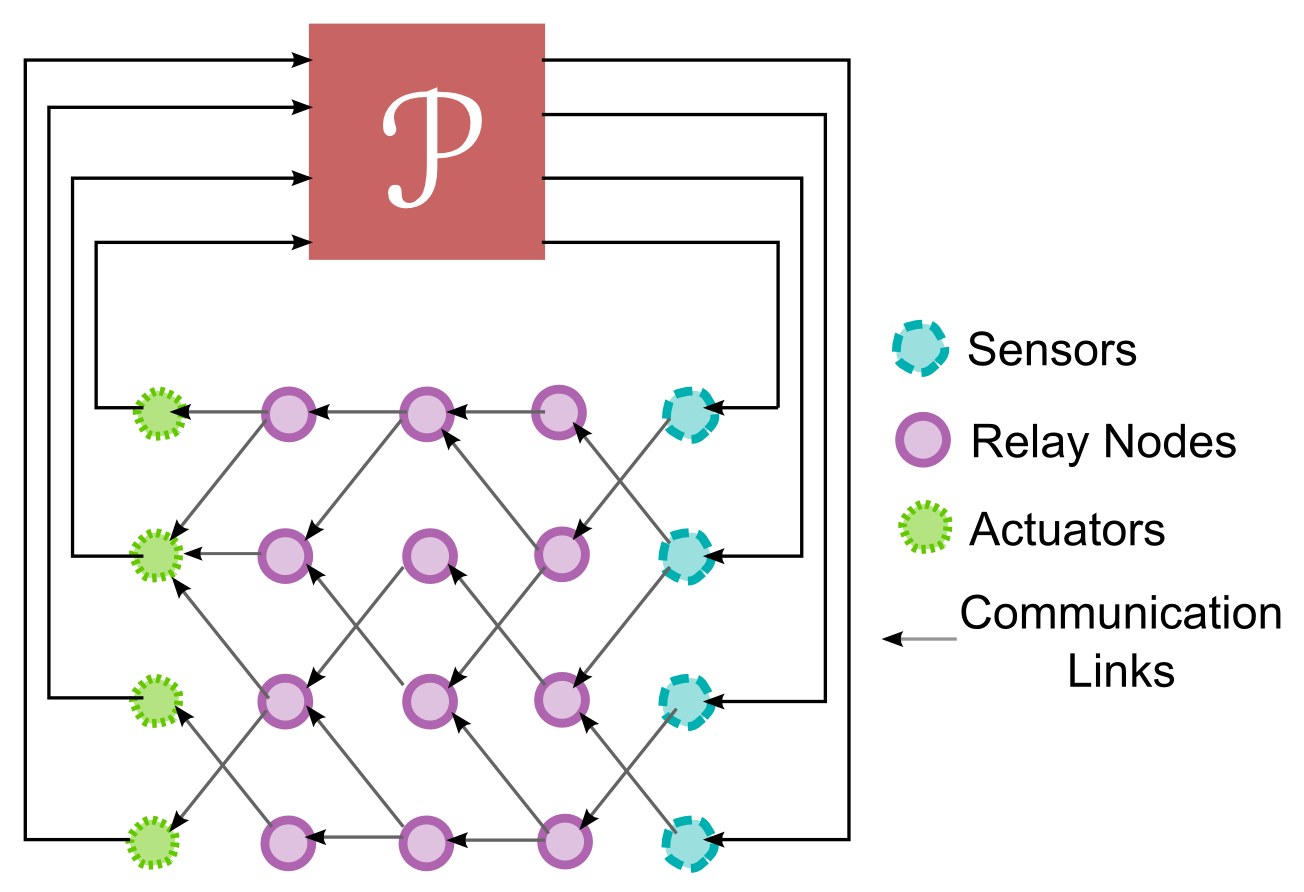}
\caption{A networked control system over a mesh network, where the
  controllers can be located on any node.}
\label{fig:dist_ctrl_mesh_net}
\end{center}
\end{figure}

To gain more insights on how to architect and design NCSs, two
limitations in the approach of many current NCS research studies need
to be addressed.  The first limitation is the use of simple models of
packet delivery over a point-to-point link or a star network topology
to represent the network, which are often multihop and more complex.
The second limitation is the treatment of the network as something
designed and fixed a priori before the design of the control system.
Very little information is passed through the interface between the
network and the control system, limiting the interaction between the
two ``layers'' to tune the controller to the network conditions, and
vice versa.

\subsection{Related Works}

Schenato et al. \cite{schenato-procIEEE:2007} and Ishii \cite{ishii:2008} study stability
and controller synthesis for different control system architectures,
but they both model networks as i.i.d. Bernoulli processes that drop
packets on a single link.  The information passed through the
interface between the network and the control system is the packet
drop probability of the link, which is assumed to be known and fixed.
Seiler and Sengupta \cite{seiler:2005} study stability and $\mathcal{H}_\infty$
controller synthesis when the network is modeled as a packet-dropping
link described by a two-state Markov chain (Gilbert-Elliott model),
where the information passed through the network-controller interface
are the transition probabilities of the Markov chain.  Elia \cite{elia:2005}
studies stability and the synthesis of a stabilizing controller when
the network is represented by an LTI system with stochastic
disturbances modeled as parallel, independent, multiplicative fading
channels.

Some related work in NCSs do use models of multihop networks.  For
instance, work on consensus of multi-agent systems
\cite{olfati-saber:2007} typically study how the connectivity
graph(s) provided by the network affects the convergence of the
system, and is not focused on modeling the links.
Robinson and Kumar \cite{robinson:2007} study the optimal placement of a controller in
a multihop network with i.i.d. Bernoulli packet-dropping links, where
the packet drop probability is known to the controller.
Gupta et al. \cite{gupta:2009} study how to optimally process and forward sensor
measurements at each node in a multihop network for optimal LQG
control, and analyze stability when packet drops on the links are
modeled as spatially-independent Bernoulli, spatially-independent
Gilbert-Elliott, or memoryless spatially-correlated
processes.\footnote{Here, ``spatially'' means ``with respect to other
  links.''}  Varagnolo et al. \cite{varagnolo-cdc:2008} compare the performance of a
time-varying Kalman filter on a wireless TDMA mesh network under
unicast routing and constrained flooding.  The network model describes
the routing topology and schedule of an implemented communication
protocol, TSMP \cite{pister:2008}, but it assumes that transmission
successes on the links are spatially-independent and memoryless.
Both Gupta et al. \cite{gupta:2009} and Varagnolo et al. \cite{varagnolo-cdc:2008} are concerned
with estimation when packet drops occur on the sensing channel, and the 
estimators do not need to know network parameters like the packet loss
probability.

\subsection{Contributions}

Our approach is a step toward using more sophisticated, multihop
network models and passing more information through the interface between
the controller and the network.  Similar to Gupta et al. \cite{gupta:2009}, we
model the network routing topology as a graph of independent links,
where transmission success on each link is described by a two-state
Markov chain.  The network model consists of the routing topology and
a global TDMA transmission schedule.  Such a minimalist network model
captures the essence of how a network with bursty links can have
correlated packet deliveries \cite{willig:2002}, which are
particularly bad for control when they result in bursts of packet losses.
Using this model, we propose a network estimator to estimate, without
loss of information, the state of the network given the past
packet deliveries.\footnote{Strictly speaking, we obtain the
probability distribution on the states of the network, not a single
point estimate.}  The network estimate is translated to a joint
probability distribution predicting the success of future packet
deliveries, which is passed through the network-controller interface
so the controller can compensate for unavoidable network outages.  The
network estimator can also be used to notify a network manager when
the network is broken and needs to be reconfigured or reprovisioned, a
direction for future research.

Section~\ref{sec:prob_form} describes our plant and network models.
We propose two network estimators, the
Static Independent links, Hop-by-hop routing, Scheduled (SIHS) network
estimator and the Gilbert-Elliott Independent links, Hop-by-hop routing, 
Scheduled (GEIHS) network estimator in Section~\ref{sec:net_est_pkt_pred}.  
Next, we design a finite-horizon, Future-Packet-Delivery-optimized 
(FPD) LQG controller to utilize the packet delivery predictions 
provided by the network estimators, presented in Section~\ref{sec:control_syn}.
Section~\ref{sec:ex_sim} provides an example and simulations
demonstrating how the GEIHS network estimator combined with the FPD
controller can provide better performance than a classical LQG
controller or a controller assuming i.i.d. packet deliveries.
Finally, Section~\ref{sec:conclusions} describes the limitations of
our approach and future work.

%% file: sections/problem_formulation.tex
\section{Problem Formulation}
\label{sec:prob_form}

\begin{figure}
\begin{center}
\includegraphics[width=8cm]{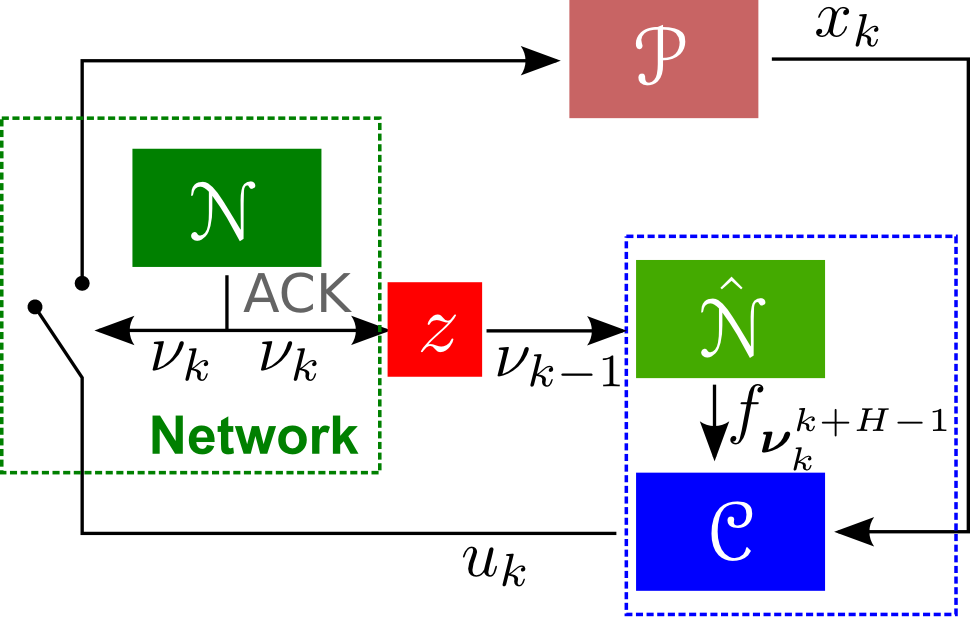}
\caption{A control loop for plant $\Pe$ with the network on the
  actuation channel. The network estimator $\hat{\Ne}$ passes packet
  delivery predictions $f_{\vect{\nu}_{^{k}}^{_{k+H-1}}}$ to the FPD
  controller $\Ce$, with past packet delivery information obtained
  from the network $\Ne$ over an acknowledgement (ACK) channel.}
\label{fig:system_diagram}
\end{center}
\end{figure}

This paper studies an instance of the general system architecture
depicted in Figure~\ref{fig:dist_ctrl_mesh_net}, with a single control
loop containing one sensor and one actuator.  One network estimator
and one controller are placed at the sensor, and we assume that an
end-to-end acknowledgement (ACK) that the controller-to-actuator
packet is delivered is always received at the network estimator, as
shown in Figure~\ref{fig:system_diagram}.  For simplicity, we assume
that the plant dynamics are significantly slower than the end-to-end
packet delivery deadline, so that we can ignore the delay introduced
by the network.  The general problem is to jointly design a network
estimator and controller that can optimally control the plant using
our proposed SIHS and GEIHS network models.  In our problem setup, the
controller is only concerned with the past, present, and future packet
delivery sequence and not with the detailed behavior of the network,
nor can it affect the behavior of the network.  Therefore, the network
estimation problem decouples from the control problem.  The
information passed through the network-controller interface is the
packet delivery sequence, specifically the joint probability
distribution describing the future packet delivery predictions.

\subsection{Plant and Network Models}
\label{sec:sys_net_model}

The state dynamics of the plant $\Pe$ in Figure~\ref{fig:system_diagram}
is given by
\begin{equation} \label{Eq:StateSpace}
x_{k+1} = A x_k + \nu_k B u_k + w_k \quad,
\end{equation}
where $A \in \R^{\ell \times \ell}$, $B \in \R^{\ell \times m}$, and
$w_k$ are i.i.d. zero-mean Gaussian random variables with covariance
matrix $R_w \in \Sbb_+^\ell$, where $\Sbb_+^\ell$ is the set of $\ell
\times \ell$ positive semidefinite matrices. The initial state $x_0$
is a zero-mean Gaussian random variable with covariance matrix $R_0
\in \Sbb_+^\ell$ and is mutually independent of $w_k$.  The binary
random variable $\nu_k$ indicates whether a packet from the controller
reaches the actuator ($\nu_k=1$) or not ($\nu_k=0$), and each $\nu_k$
is independent of $x_0$ and $w_k$ (but the $\nu_k$'s are not
independent of each other).

Let the discrete sampling times for the control system be indexed by
$k$, but let the discrete time for schedule time slots (described
below) be indexed by $t$.  The time slot intervals are smaller than
the sampling intervals.  The time slot when the control packet at
sample time $k$ is generated is denoted $t_k$, and the deadline for
receiving the control packet at the receiver is $t_k'$.  We assume
that $t_k' \le t_{k+1}$ for all $k$.  Figure~\ref{fig:time}
illustrates the relationship between $t$ and $k$.

\begin{figure}
\begin{center}
\includegraphics[width=10cm]{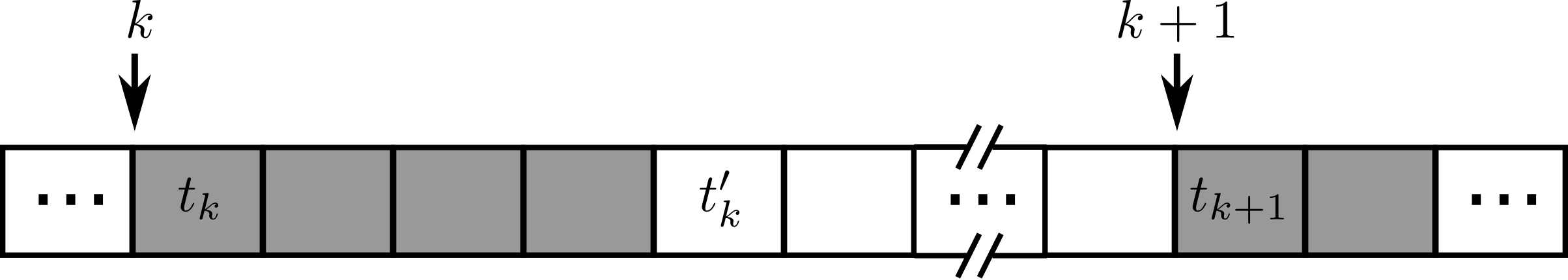}
\caption{The packet containing the control input $u_k$ is generated 
right before time slot $t_k$.  The packet may be in transit through the 
network in the shaded time slots, until right before time slot $t_k'$.  
Thus, time $t_k$ is aligned with the beginning of the time slot.}
\label{fig:time}
\end{center}
\end{figure}

The model of the TDMA wireless mesh network ($\Ne$ in
Figure~\ref{fig:system_diagram}) consists of a routing topology $G$,
a link model describing how the transmission success of a link evolves
over time, and a fixed repeating schedule $\mathbf{F}^{(T)}$.  The
SIHS network model and the GEIHS network model only differ in the link
model.  Each of these components will be described in detail below.

The routing topology is described by $G = (\Vc,\Ec)$, a connected
directed acyclic graph with the set of vertices (nodes) $\Vc =
\{1,\dots,M\}$ and the set of directed edges (links) $\Ec \subseteq
\{(i,j) \: : \: i,j \in \Vc, \: i\ne j\}$, where the number of edges is
denoted $E$.  The source node is denoted $a$ and the sink
(destination) node is denoted $b$.  Only the destination node has no
outgoing edges.

At any moment in time, the links in $G$ can be either be up (succeeds
if attempt to transmit packet) or down (fails if attempt to transmit
packet).  Thus, there are $2^E$ possible topology realizations
$\tilde{G} = (\Vc,\tilde{\Ec})$, where $\tilde{\Ec} \subseteq \Ec$
represents the edges that are up.\footnote{Symbols with a tilde
($\tilde{\cdot}$) denote values that can be taken on by random
variables, and can be the arguments to probability distribution
functions (pdfs).}

At time $t_k$, the actual state of the topology is one of the topology
realizations but it is not known to the network estimator.  With some
abuse of terminology, we define $G^{_{(k)}}$ to be the random variable
representing the state of the topology at time
$t_k$.\footnote{Strictly speaking, $G^{_{(k)}}$ is a function mapping
events to the set of all topology realizations, not to the set of
real numbers.}

This paper considers the network under two link models, the static
link model and the Gilbert-Elliott (G-E) link model.  Both network
models assume all the links in the network are independent.

The static link model assumes the links do not switch between being 
up and down while packets are sent through the network.  Therefore, the 
sequence of topology realizations over time is constant.  While not 
realistic, it leads to the simple network estimator in 
Section~\ref{sec:sl_net_est_pkt_pred} for pedagogical purposes.  The
a priori transmission success probability of link $l = (i,j)$ is $p_l$.

The G-E link model represents each link $l$ by the two-state
Markov chain shown in Figure~\ref{fig:G_E_model}.  At each sample
time $k$, a link in state 0 (down) transitions to state 1 (up) with
probability $p_l^\mathrm{u}$, and a link from state 1 transitions to
state 0 with probability $p_l^\mathrm{d}$.\footnote{We can easily
instead use a G-E link model that advances at each time step $t$,
but it would make the following exposition and notation more
complicated.}  The steady-state probability of being in state 1, which
we use as the a priori probability of the link being up, is
\[
p_l = p_l^\mathrm{u} / (p_l^\mathrm{u} + p_l^\mathrm{d})
\quad .
\]

\begin{figure}
\begin{center}
\includegraphics[width=5cm]{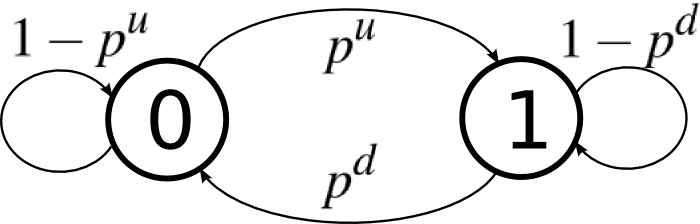}
\caption{Gilbert-Elliott link model}
\label{fig:G_E_model}
\end{center}
\end{figure}

The fixed, repeating schedule of length $T$ is represented by a
sequence of matrices $\mathbf{F}^{(T)} =
(F^{(1)},F^{(2)},\dots,F^{(T)})$, where the matrix $F^{(t - 1
\pmod{T} + 1)}$ represents the links scheduled at time $t$.  The
matrix $F^{(t)} \in \{0,1\}^{M \times M}$ is defined from the set
$\Fc^{(t)} \subseteq \Ec$ containing the links scheduled for
transmission at time $t$.  We assume that nodes can only unicast
packets, meaning that for all nodes $i$, if $(i,j) \in \Fc^{(t)}$ then
for all $v \ne j, (i,v) \not\in \Fc^{(t)}$.  Furthermore, a node
holds onto a packet if the transmission fails and can retransmit the
packet the next time an outgoing link is scheduled (hop-by-hop
routing).  Thus, the matrix $F^{(t)}$ has entries
\[
F_{ij}^{(t)} = \begin{cases}
    1 & \text{if $(i,j) \in \Fc^{(t)}$, or}\\
      & \text{if $i = j$ and $\forall v \in \Vc, (i,v) \not\in \Fc^{(t)}$}\\
    0 & \text{otherwise.}
    \end{cases}
\]

An exact description of the network consists of the sequence of
topology realizations over time and the schedule $\mathbf{F}^{(T)}$.
Assuming a topology realization $\tilde{G}$, the links that are
scheduled and up at any given time $t$ are represented by the matrix
$\tilde{F}^{(t;\tilde{G})} \in \{0,1\}^{M \times M}$, with entries
\begin{equation}
\label{eq:net_realize_matrix}
\tilde{F}_{ij}^{(t;\tilde{G})} = \begin{cases}
    1 & \text{if $(i,j) \in \Fc^{(t)} \cap \tilde{\Ec}$, or }\\
      & \text{if $i = j$  and $\forall v \in \Vc, (i,v) \not\in
        \Fc^{(t)} \cap \tilde{\Ec}$}\\
    0 & \text{otherwise.}
    \end{cases}
\end{equation}

Define the matrix $\tilde{F}^{(t,t';\tilde{G})} =
\tilde{F}^{(t;\tilde{G})}\tilde{F}^{(t+1;\tilde{G})}\dotsm
\tilde{F}^{(t';\tilde{G})}$, such that entry
$\tilde{F}_{ij}^{(t,t';\tilde{G})}$ is 1 if a packet at node $i$ at
time $t$ will be at node $j$ at time $t'$, and is 0 otherwise.  Since
the destination $b$ has no outgoing links, a packet sent from the
source $a$ at time $t$ reaches the destination $b$ at or before time
$t'$ if and only if $\tilde{F}_{ab}^{(t,t';\tilde{G})} = 1$.  To
simplify the notation, let the function $\delta_\kappa$ indicate
whether the packet delivery $\tilde{\nu} \in \{0,1\}$ is consistent
with the topology realization $\tilde{G}$, assuming the packet was
generated at $t_\kappa$, i.e.,
\begin{equation}
\label{eq:indicate_pkt_dlvry_match}
  \delta_\kappa(\tilde{\nu};\tilde{G}) = \begin{cases}
      1 & \text{if $\tilde{\nu} = \tilde{F}_{ab}^{(t_\kappa,t_\kappa';\tilde{G})}$}\\
      0 & \text{otherwise.}
    \end{cases}
\end{equation}
The function assumes the fixed repeating schedule $\mathbf{F}^{(T)}$,
the packet generation time $t_\kappa$, the deadline $t_\kappa'$, the
source $a$, and the destination $b$ are implicitly known.

\subsection{Network Estimators}
\label{sec:net_estimator}

As shown in Figure~\ref{fig:system_diagram}, at each sample time $k$ the
network estimator $\hat{\Ne}$ takes as input the previous packet
delivery $\nu_{k-1}$, estimates the topology realization using the
network model and all past packet deliveries, and outputs the joint
probability distribution of future packet deliveries
$f_{\vect{\nu}_{^{k}}^{_{k+H-1}}}$.  For clarity in the following
exposition, let $\Ve_\kappa \in \{0,1\}$ be the value taken on by the
packet delivery random variable $\nu_\kappa$ at some past sample
time $\kappa$.  Let the vector $\vect{\Ve}_{^0}^{_{k-1}} = [\Ve_0,
\dots, \Ve_{k-1}]$ denote the history of packet deliveries at sample
time $k$, the values taken on by the vector of random variables
$\vect{\nu}_{^0}^{_{k-1}} = [\nu_0, \dots, \nu_{k-1}]$.  Then,
\begin{equation}
\label{eq:pkt_pred}
f_{\vect{\nu}_{^{k}}^{_{k+H-1}}}
(\vect{\tilde{\nu}}_{^0}^{_{H-1}}) =
  \Pr( \vect{\nu}_{^{k}}^{_{k+H-1}} =\vect{\tilde{\nu}}_{^0}^{_{H-1}} |
      \vect{\nu}_{^0}^{_{k-1}} = \vect{\Ve}_{^0}^{_{k-1}})
\end{equation}
is the prediction of the next $H$ packet deliveries, where
$\vect{\nu}_{^{k}}^{_{k+H-1}} = [\nu_{k}, \dots, \nu_{k+H-1}]$ is a
vector of random variables representing future packet deliveries and
$\vect{\tilde{\nu}}_{^0}^{_{H-1}} \in \{0,1\}^H$.

The SIHS and GEIHS network estimators only differ in the network
models.  The parameters of the network models ---
topology $G$, schedule $\mathbf{F}^{(T)}$, link probabilities
$\{p_l\}_{l \in \Ec}$ or $\{p_l^\mathrm{u},p_l^\mathrm{d}\}_{l \in
\Ec}$, source $a$, sink $b$, packet generation times $t_k$, and
deadlines $t_k'$ --- are known a priori to the network estimators
and are left out of the conditional probability expressions.

In Section~\ref{sec:net_est_pkt_pred}, we will use the probability
distribution on the topology realizations (our network state estimate),
\[
\Pr(G^{_{(k)}} = \tilde{G} | \vect{\nu}_{^0}^{_{k-1}} =
     \vect{\Ve}_{^0}^{_{k-1}}) \quad,
\]
to obtain $f_{\vect{\nu}_{^{k}}^{_{k+H-1}}}$ from
$\vect{\Ve}_{^0}^{_{k-1}}$ and the network model.

\subsection{FPD Controller}
\label{sec:ctrlr_pkt_pred}

The FPD controller ($\Ce$ in Figure~\ref{fig:system_diagram}) optimizes
the control signals to the statistics of the future packet delivery
sequence, derived from the past packet delivery sequence. We choose
the optimal control framework because the cost function allows us to
easily compare the FPD controller with other controllers. The control
policy operates on the information set
\begin{equation} \label{Eq:Ic}
\Ic^{^\Ce}_k = \{ \vect{x}_{^{0}}^{_{k}}, \vect{u}_{^{0}}^{_{k-1}},
                  \vect{\nu}_{^{0}}^{_{k-1}} \} \quad .
\end{equation}
The control policy minimizes the linear quadratic cost function
\begin{equation} \label{Eq:LQGCriterion}
\E\left[ \Sst x_N^T Q_0
     x_N+\sum_{n=0}^{N-1} x_n^T Q_1 x_n + \nu_n u_n^T Q_2 u_n \D \right]
\quad ,
\end{equation}
where $Q_0$, $Q_1$, and $Q_2$ are positive definite weighting matrices
and $N$ is the finite horizon, to get the minimum cost
\[
J = \min_{_{u_0,\dots,u_{N-1}}}\E\left[ \Sst x_N^T Q_0
     x_N+\sum_{n=0}^{N-1} x_n^T Q_1 x_n + \nu_n u_n^T Q_2 u_n \D \right]
\quad .
\]
Section \ref{sec:control_syn} will show that the resulting
architecture separates into a network estimator and a controller which
uses the pdf $f_{\vect{\nu}_{^{k}}^{_{k+H-1}}}$ supplied by the
network estimator ($\hat{\Ne}$ in Figure~\ref{fig:system_diagram}) to
find the control signals $u_k$.

%% file: sections/network_estimator.tex
\section{Network Estimation and Packet Delivery Prediction}
\label{sec:net_est_pkt_pred}

We will use recursive Bayesian estimation to estimate the state of the
network, and use the network state estimate to predict future packet
deliveries.  Figure~\ref{fig:net_est_graphical_model} is the graphical
model / hidden Markov model \cite{smyth:1997} describing our
recursive estimation problem.

\begin{figure}
\begin{center}
\includegraphics[width=10cm]{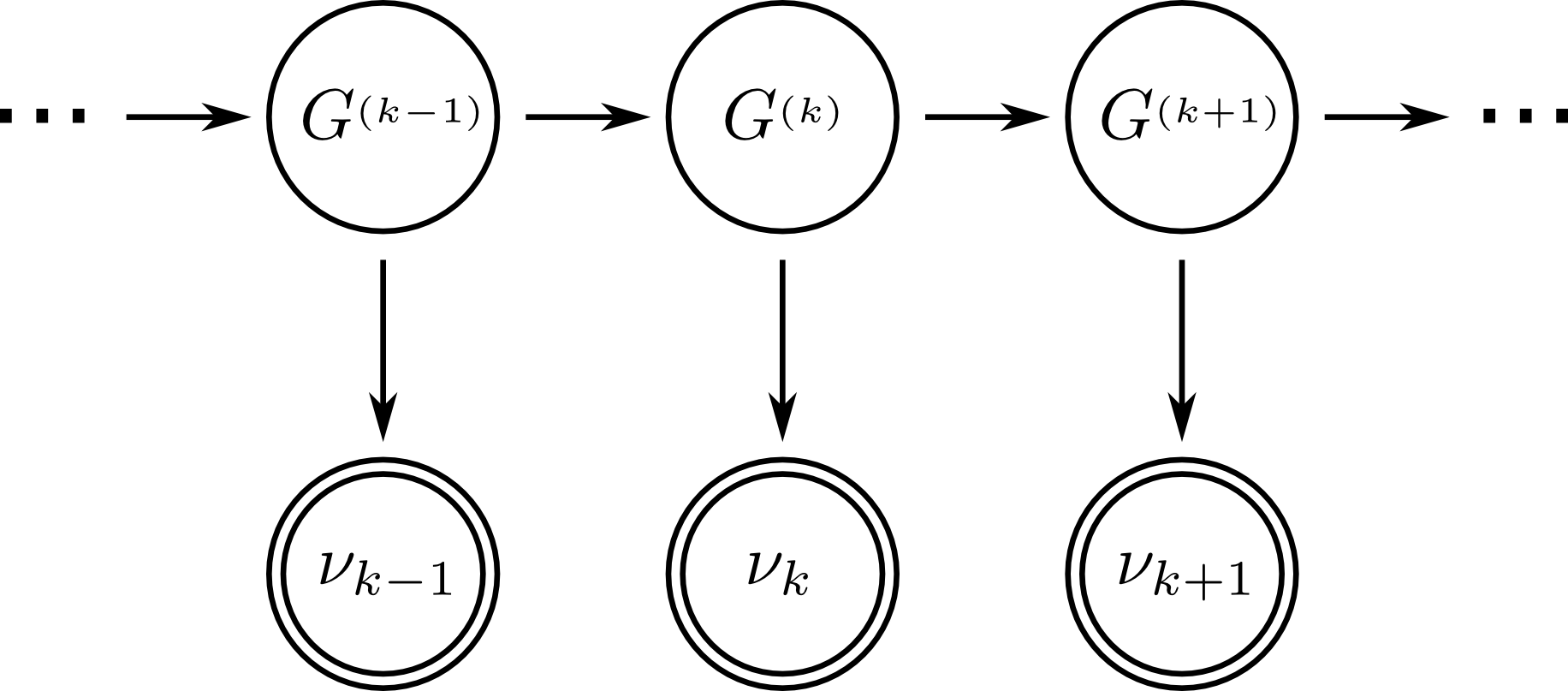}
\caption{Graphical model describing the network estimation problem.
  $\nu_k$ is the measurement output variable at time $k$, and
  $G^{_{(k)}}$ is the hidden state of the network.}
\label{fig:net_est_graphical_model}
\end{center}
\end{figure}

\subsection{SIHS Network Estimator}
\label{sec:sl_net_est_pkt_pred}

The steps in the SIHS network estimator are derived from
\eqref{eq:pkt_pred}.  We introduce new notation for conditional pdfs
(i.e., $\alpha_k,\beta_k,Z_k$), which will be used later to state the
steps in the estimator compactly.\footnote{A semicolon is used in the
conditional pdfs to separate the values being conditioned on from
the remaining arguments.}  First, express
$f_{\vect{\nu}_{^{k}}^{_{k+H-1}}}(\vect{\tilde{\nu}}_{^0}^{_{H-1}})$ as
\[
 \underbrace{\Pr(\vect{\nu}_{^{k}}^{_{k+H-1}} = \vect{\tilde{\nu}}_{^0}^{_{H-1}} |
                 \vect{\nu}_{^0}^{_{k-1}} = \vect{\Ve}_{^0}^{_{k-1}})}_
            {f_{\vect{\nu}_{^{k}}^{_{k+H-1}}}(\vect{\tilde{\nu}}_{^0}^{_{H-1}})} =
  \sum_{\tilde{G}}
   \underbrace{\Pr(\vect{\nu}_{^{k}}^{_{k+H-1}} = \vect{\tilde{\nu}}_{^0}^{_{H-1}} |
                   G^{_{(k-1)}} = \tilde{G})}_
              {\alpha_k(\vect{\tilde{\nu}}_{^0}^{_{H-1}};\tilde{G})} \cdot
   \underbrace{\Pr(G^{_{(k-1)}} = \tilde{G} | \vect{\nu}_{^0}^{_{k-1}} = \vect{\Ve}_{^0}^{_{k-1}})}_
              {\beta_k(\tilde{G})}
\quad ,
\]
where we use the relation
\[
  \Pr(\vect{\nu}_{^{k}}^{_{k+H-1}} = \vect{\tilde{\nu}}_{^0}^{_{H-1}} |
      G^{_{(k-1)}} = \tilde{G},\vect{\nu}_{^0}^{_{k-1}} = \vect{\Ve}_{^0}^{_{k-1}}) =
  \Pr(\vect{\nu}_{^{k}}^{_{k+H-1}} = \vect{\tilde{\nu}}_{^0}^{_{H-1}} |
      G^{_{(k-1)}} = \tilde{G}) \quad .
\]
This relation states that given the state of the network, future
packet deliveries are independent of past packet deliveries.  The
expression $\Pr(\vect{\nu}_{^{k}}^{_{k+H-1}} =
\vect{\tilde{\nu}}_{^0}^{_{H-1}} | G^{_{(k-1)}} = \tilde{G})$
indicates whether the future packet delivery sequence
$\vect{\tilde{\nu}}_{^0}^{_{H-1}}$ is consistent with the graph
realization $\tilde{G}$, meaning
\[
\Pr(\vect{\nu}_{^{k}}^{_{k+H-1}} = \vect{\tilde{\nu}}_{^0}^{_{H-1}} |
G^{_{(k-1)}} = \tilde{G}) = \prod_{h=0}^{H-1} \delta_{k+h}(\tilde{\nu}_h;\tilde{G})
\quad,
\]
where $\prod$ is the \emph{and} operator (sometimes denoted
$\bigwedge$). The network state estimate at sample time $k$ from past packet
deliveries is $\beta_k(\tilde{G})$ and is obtained from the network state
estimate at sample time $k-1$, since
\begin{equation}
\label{eq:sl_net_est_derivation}
\underbrace{\Pr(G^{_{(k-1)}} = \tilde{G} | \vect{\nu}_{^0}^{_{k-1}} = \vect{\Ve}_{^0}^{_{k-1}})}_
           {\beta_k(\tilde{G})} =
\frac{\overbrace{\Pr(\nu_{k-1} = \Ve_{k-1} | G^{_{(k-1)}} = \tilde{G})}^{\delta_{k-1}(\Ve_{k-1};\tilde{G})}
      \cdot
      \overbrace{\Pr(G^{_{(k-1)}} = \tilde{G} | \vect{\nu}_{^0}^{_{k-2}} = \vect{\Ve}_{^0}^{_{k-2}})}^
                {\beta_{k-1}(\tilde{G})}}
     {\underbrace{\Pr(\nu_{k-1} = \Ve_{k-1} | \vect{\nu}_{^0}^{_{k-2}} = \vect{\Ve}_{^0}^{_{k-2}})}_{Z_k}}
     \quad.
\end{equation}
Here, $\Pr(G^{_{(k-1)}} = \tilde{G} | \vect{\nu}_{^0}^{_{k-2}} = \vect{\Ve}_{^0}^{_{k-2}}) =
\Pr(G^{_{(k-2)}} = \tilde{G} | \vect{\nu}_{^0}^{_{k-2}} = 
\vect{\Ve}_{^0}^{_{k-2}}) = \beta_{k-1}(\tilde{G})$ for the static
link model because $G^{_{(k-1)}} = G^{_{(k-2)}} = G^{_{(0)}}$.  Again,
we used the independence of future packet deliveries from past packet
deliveries given the network state,
\[
\Pr(\nu_{k-1} = \Ve_{k-1} | G^{_{(k-1)}} = \tilde{G} ,\vect{\nu}_{^0}^{_{k-2}} = \vect{\Ve}_{^0}^{_{k-2}}) =
\Pr(\nu_{k-1} = \Ve_{k-1} | G^{_{(k-1)}} = \tilde{G}) \quad.
\]
Note that $\Pr(\nu_{k-1} = \Ve_{k-1} | G^{_{(k-1)}} = \tilde{G})$ can
only be 0 or 1, indicating whether the packet delivery is consistent
with the graph realization.  Finally, $\Pr(\nu_{k-1} = \Ve_{k-1} |
\vect{\nu}_{^0}^{_{k-2}} = \vect{\Ve}_{^0}^{_{k-2}})$ is the same for
all $\tilde{G}$, so it is treated as a normalization constant.

At sample time $k=0$, when no packets have been sent through the
network, $\beta_0(\tilde{G}) = \Pr(G^{_{(0)}} = \tilde{G})$, which is
expressed in \eqref{eq:sl_init_net_est} below.  This equation comes
from the assumption that all links in the network are independent.

\begin{subequations}
\label{eqs:sihs_net_est}
To summarize, the \textit{SIHS Network Estimator and Packet Delivery
Predictor} is a recursive Bayesian estimator where the measurement
output step consists of
\begin{align}
\label{eq:sl_pkt_pred_update}
  f_{\vect{\nu}_{^{k}}^{_{k+H-1}}}(\vect{\tilde{\nu}}_{^0}^{_{H-1}}) &=
  \sum_{\tilde{G}} \alpha_k(\vect{\tilde{\nu}}_{^0}^{_{H-1}};\tilde{G}) 
                   \cdot \beta_k(\tilde{G}) \\
\label{eq:sl_pkt_k_pred}
  \alpha_k(\vect{\tilde{\nu}}_{^0}^{_{H-1}};\tilde{G}) &= 
    \prod_{h=0}^{H-1} \delta_{k+h}(\tilde{\nu}_h;\tilde{G}) \quad ,
\end{align}
and the innovation step
consists of
\begin{align}
\label{eq:sl_net_est_update}
  \beta_k(\tilde{G}) &= \frac{\delta_{k-1}(\Ve_{k-1};\tilde{G}) \cdot
    \beta_{k-1}(\tilde{G})}{Z_k}\\
\label{eq:sl_init_net_est}
\beta_0(\tilde{G}) &= \left(\prod_{l \in \tilde{\Ec}} p_l \right)
\left(\prod_{l \in \Ec \backslash \tilde{\Ec}} 1-p_l \right)
\quad ,
\end{align}
where $\alpha_k$ and $\beta_k$ are functions, $Z_k$ is a normalization
constant such that $\sum_{\tilde{G}} \beta_k(\tilde{G}) = 1$, and the
functions $\delta_{k+h}$ and $\delta_{k-1}$ are defined by
\eqref{eq:indicate_pkt_dlvry_match}.
\end{subequations}

\subsection{GEIHS Network Estimator}
\label{sec:ge_net_est_pkt_pred}

For compact notation in the probability expressions below, we use 
$\vect{\Ve}_{^0}^{_{k-1}}$ in place of
$\vect{\nu}_{^0}^{_{k-1}} = \vect{\Ve}_{^0}^{_{k-1}}$ and only write
the random variable and not its value ($\tilde{\cdot}$) .

The derivation of the GEIHS network estimator is similar to the
previous derivation, except that the state of the network evolves with
every sample time $k$.  Since all the links in the network are
independent, the probability that a given topology $\tilde{G}'$ at
sample time $k-1$ transitions to a topology $\tilde{G}$ after one sample
time is given by
\begin{equation}
\label{eq:ge_topo_evolve}
  \Gamma(\tilde{G};\tilde{G}') =
  \Pr(G^{_{(k)}} | G^{_{(k-1)}}) =
   \left(\prod_{l_1 \in \tilde{\Ec}' \cap \tilde{\Ec}}
		  1 - p_{l_1}^\mathrm{d} \right)
    \left(\prod_{l_2 \in \tilde{\Ec}' \backslash \tilde{\Ec}}
    	      p_{l_2}^\mathrm{d} \right)
    \left(\prod_{l_3 \in \tilde{\Ec} \backslash \tilde{\Ec}'}
    	      p_{l_3}^\mathrm{u} \right)	
    \left(\prod_{l_4 \in \Ec \backslash (\tilde{\Ec}' \cup \tilde{\Ec})}
    	      1 - p_{l_4}^\mathrm{u} \right)
\quad .
\end{equation}

First, express
$f_{\vect{\nu}_{^{k}}^{_{k+H-1}}}(\vect{\tilde{\nu}}_{^0}^{_{H-1}})$ as
\[
\underbrace{\Pr(\vect{\nu}_{^{k}}^{_{k+H-1}} | \vect{\Ve}_{^0}^{_{k-1}})}_{
  f_{\vect{\nu}_{^{k}}^{_{k+H-1}}}(\vect{\tilde{\nu}}_{^0}^{_{H-1}})} =
  \sum_{G^{_{(k+H-1)}}} 
  \underbrace{\Pr(\nu_{k+H-1} | G^{_{(k+H-1)}})}_{
              \delta_{k+H-1}(\tilde{\nu}_{H-1};\tilde{G}_{H-1})} \cdot
  \underbrace{\Pr(\vect{\nu}_{^{k}}^{_{k+H-2}},G^{_{(k+H-1)}} |
              \vect{\Ve}_{^0}^{_{k-1}})}_{
              \alpha_{H-1|k-1}(\vect{\tilde{\nu}}_{^0}^{_{H-2}},\tilde{G}_{H-1}) }
\quad ,
\]
where for $h=2,\dots,H-1$
\begin{equation}
\label{eq:ge_pkt_k_pred_derivation}
  \underbrace{\Pr(\vect{\nu}_{^{k}}^{_{k+h-1}},G^{_{(k+h)}} |
  \vect{\Ve}_{^0}^{_{k-1}})}_{
  \alpha_{h|k-1}(\vect{\tilde{\nu}}_{^0}^{_{h-1}},\tilde{G}_h)} =
      \sum_{G^{_{(k+h-1)}}}
      \underbrace{\Pr(G^{_{(k+h)}} | G^{_{(k+h-1)}})}_{
                      \Gamma(\tilde{G}_h;\tilde{G}_{h-1})} \cdot
      \underbrace{\Pr(\nu_{k+h-1} | G^{_{(k+h-1)}})}_{
                  \delta_{k+h-1}(\tilde{\nu}_{h-1};\tilde{G}_{h-1})} \cdot
      \underbrace{\Pr(\vect{\nu}_{^{k}}^{_{k+h-2}},G^{_{(k+h-1)}} |
                  \vect{\Ve}_{^0}^{_{k-1}})}_{
                  \alpha_{h-1|k-1}(\vect{\tilde{\nu}}_{^0}^{_{h-2}},\tilde{G}_{h-1})}
      \quad .
\end{equation}
When $h=1$, replace $\Pr(\vect{\nu}_{^{k}}^{_{k+h-2}},G^{_{(k+h-1)}}
| \vect{\Ve}_{^0}^{_{k-1}})$ in \eqref{eq:ge_pkt_k_pred_derivation} with
$\Pr(G^{_{(k)}} | \vect{\Ve}_{^0}^{_{k-1}}) =
\beta_{k|k-1}(\tilde{G})$.  The value $\beta_{k|k-1}(\tilde{G})$ comes
from
\[
  \underbrace{\Pr(G^{_{(k)}} | \vect{\Ve}_{^0}^{_{k-1}})}_{
              \beta_{k|k-1}(\tilde{G})} =
  \sum_{G^{_{(k-1)}}}
  \underbrace{\Pr(G^{_{(k)}} | G^{_{(k-1)}})}_{
              \Gamma(\tilde{G};\tilde{G}')} \cdot
  \underbrace{\Pr(G^{_{(k-1)}} | \vect{\Ve}_{^0}^{_{k-1}})}_{
              \beta_{k-1|k-1}(\tilde{G}')} \quad .
\]
The value $\beta_{k-1|k-1}(\tilde{G}')$ comes from
\eqref{eq:sl_net_est_derivation}, with $\beta_{k}$ replaced by
$\beta_{k-1|k-1}$ and $\beta_{k-1}$ replaced by $\beta_{k-1|k-2}$.  Finally,
$\beta_{0|-1}(\tilde{G}) = \Pr(G^{_{(0)}})$, where all links are
independent and have link probabilities equal to their steady-state
probability of being in state 1, and is expressed in
\eqref{eq:ge_init_net_est} below.

\begin{subequations}
\label{eqs:geihs_net_est}
To summarize, the \textit{GEIHS Network Estimator and Packet Delivery
Predictor} is a recursive Bayesian estimator.  The measurement
output step consists of
\begin{equation}
\label{eq:ge_pkt_pred_update}
  f_{\vect{\nu}_{^{k}}^{_{k+H-1}}}(\vect{\tilde{\nu}}_{^0}^{_{H-1}}) =
  \sum_{\tilde{G}_{H-1}} \delta_{k+H-1}(\tilde{\nu}_{H-1};\tilde{G}_{H-1}) \cdot
  \alpha_{H-1|k-1}(\vect{\tilde{\nu}}_{^0}^{_{H-2}},\tilde{G}_{H-1})
  \quad ,
\end{equation}
where the function $\alpha_{H-1|k-1}$ is obtained from the following
recursive equation for $h=2,\ldots,H-1$:
\begin{equation}
\label{eq:ge_pkt_k_pred}
  \alpha_{h|k-1}(\vect{\tilde{\nu}}_{^0}^{_{h-1}},\tilde{G}_h) =
    \sum_{\tilde{G}_{h-1}} \Gamma(\tilde{G}_h;\tilde{G}_{h-1}) \cdot
     \delta_{k+h-1}(\tilde{\nu}_{h-1};\tilde{G}_{h-1}) \cdot
     \alpha_{h-1|k-1}(\vect{\tilde{\nu}}_{^0}^{_{h-2}},\tilde{G}_{h-1})
\quad ,
\end{equation}
with initial condition
\begin{equation}
\label{eq:ge_init_pkt_k_pred}
  \alpha_{1|k-1}(\vect{\tilde{\nu}}_{^0}^{_0},\tilde{G}_1) =
    \sum_{\tilde{G}} \Gamma(\tilde{G}_1;\tilde{G}) \cdot
    \delta_k(\tilde{\nu}_0;\tilde{G}) \cdot
    \beta_{k|k-1}(\tilde{G}) \quad .
\end{equation}
The prediction and innovation steps consist of
\begin{align}
\label{eq:ge_net_est_pred}
  \beta_{k|k-1}(\tilde{G}) &= \sum_{\tilde{G}'} \Gamma(\tilde{G};\tilde{G}') \cdot
	\beta_{k-1|k-1}(\tilde{G}')\\
\label{eq:ge_net_est_innov}
  \beta_{k-1|k-1}(\tilde{G}) &= \frac{\delta_{k-1}(\Ve_{k-1};\tilde{G}) \cdot
    \beta_{k-1|k-2}(\tilde{G})}{Z_{k-1}}\\
\label{eq:ge_init_net_est}
 \beta_{0|-1}(\tilde{G}) &= \left(\prod_{l \in \tilde{\Ec}} p_l \right)
 \left(\prod_{l \in \Ec \backslash \tilde{\Ec}} 1-p_l \right)
 \quad ,
\end{align}
where $\alpha_{h|k-1}$, $\beta_{k-1|k-1}$, and $\beta_{k|k-1}$ are
functions, $Z_{k-1}$ is a normalization constant such that
$\sum_{\tilde{G}} \beta_{k-1|k-1}(\tilde{G}) = 1$, and the functions
$\delta_\kappa$ (for the different values of $\kappa$ above) and
$\Gamma$ are defined by \eqref{eq:indicate_pkt_dlvry_match} and
\eqref{eq:ge_topo_evolve}, respectively.
\end{subequations}

\subsection{Packet Predictor Complexity}
\label{sec:net_est_comp_complex}

The network estimators are trying to estimate network parameters using
measurements collected at the border of the network, a general problem
studied in the field of network tomography \cite{castro:2004} under
various problem setups.  One of the greatest challenges in network
tomography is getting good estimates with low computational complexity
estimators.

Our proposed network estimators are ``optimal'' with respect to our
models in the sense that there is no loss of information, but they are
computationally expensive.

\begin{property}
\label{prop:net_est_comp_complex}
The SIHS network estimator described by the set of equations
\eqref{eqs:sihs_net_est} takes $O(E 2^E)$ to initialize and $O(2^H
2^E)$ to update the network state estimate and predictions at each step.
The GEIHS network estimator described by the set of equations
\eqref{eqs:geihs_net_est} takes $O(E 2^{2E})$ to initialize and $O(2^H
2^{2E})$ for each update step.
\end{property}

\begin{proof}
Let $D = \max_k t_k'-t_k$.  We assume that converting $\tilde{G}$ to
the set of links that are up, $\tilde{\Ec}$, takes constant time.
Also, one can simulate the path of a packet by looking up the
scheduled and successful link transmissions instead of multiplying
matrices to evaluate
$\tilde{F}_{ab}^{(t_\kappa,t_\kappa';\tilde{G})}$, so computing
$\delta_\kappa$ for each graph $\tilde{G}$ only takes $O(D)$.  The
computational complexities below assume that the pdfs can be
represented by matrices, and multiplying an $\ell \times m$ matrix
with a $m \times n$ matrix takes $O(\ell mn)$.

\noindent\textit{SIHS packet delivery predictor complexity:}

Computing $\alpha_k$ in \eqref{eq:sl_pkt_k_pred} takes $O(D 2^H 2^E)$
and computing $\beta_k$ in \eqref{eq:sl_net_est_update} takes $O(D
2^E)$, since there are $2^E$ graphs and $2^H$ packet delivery
prediction sequences.  Computing $f_{\vect{\nu}_{^k}^{_{k+H-1}}}$ in
\eqref{eq:sl_pkt_pred_update} takes $O(D 2^H 2^E)$.   The SIHS packet
delivery predictor update step is the aggregate of all these
computations and takes $O(D 2^H 2^E)$.

The initialization step of the SIHS packet delivery predictor is just
computing $\beta_0$ in \eqref{eq:sl_init_net_est}, which takes $O(E
2^E)$.

\noindent\textit{GEIHS packet delivery predictor complexity:}

Computing $\alpha_{h|k-1}$ in \eqref{eq:ge_pkt_k_pred} takes $O(D2^E +
2^E2^h + 2^{2E} 2^h)$ and computing $\alpha_{1|k-1}$ in
\eqref{eq:ge_init_pkt_k_pred} takes $O(D2^E + 2^{E+1}+2^{2E+1})$, so
computing all of them takes $O(DH2^E +2^H(2^E+2^{2E}))$, or just
$O(DH2^E + 2^H2^{2E})$.  Computing $\beta_{k|k-1}$ in
\eqref{eq:ge_net_est_pred} takes $O(2^{2E})$, and computing
$\beta_{k-1|k-1}$ in \eqref{eq:ge_net_est_innov} takes $O(D 2^E)$.
Computing $f_{\vect{\nu}_{^k}^{_{k+H-1}}}$ in
\eqref{eq:ge_pkt_pred_update} takes $O(D 2^E + 2^H 2^E)$.  The GEIHS
packet delivery predictor update step is the aggregate of all these
computations and takes $O(DH2^E + 2^H2^{2E})$.

Computing $\Gamma$ in \eqref{eq:ge_topo_evolve} takes $O(E 2^{2E})$,
and computing $\beta_{0|-1}$ in \eqref{eq:ge_init_net_est} takes $O(E
2^E)$.  The initialization step of the GEIHS packet delivery predictor
is the aggregate of these computations and takes $O(E 2^{2E})$.

If we assume that the deadline $D$ is short enough to be considered
constant, we get the computational complexities given in
Property~\ref{prop:net_est_comp_complex}.
\end{proof}

A good direction for future research is to find lower complexity,
``suboptimal'' network estimators for our problem setup, and compare them
to our ``optimal'' network estimators.

\subsection{Discussion}
\label{sec:net_est_discuss}

Our network estimators can easily be extended to incorporate
additional observations besides past packet deliveries, such as the
packet delay and packet path traces.  The latter can be obtained by
recording the state of the links that the packet has tried to traverse
in the packet payload.  The function $\delta_{k-1}$ in
\eqref{eq:sl_net_est_update} and \eqref{eq:ge_net_est_innov} just
needs to be replaced with another function that returns 1 if the the
received observation is consistent with a network topology $\tilde{G}$,
and 0 otherwise.  The advantage of using more observations than the
one bit of information provided by a packet delivery is that it will
help the GEIHS network estimator more quickly detect changes in the
network state.  A more non-trivial extension of the GEIHS network
estimator would use additional observations provided by packets from
other flows (not from our controller) to help estimate the network
state, which could significantly decrease the time for the network
estimator to detect a change in the state of the network.  This is
non-trivial because the network model would now have to account for
queuing at nodes in the network, which is inevitable with multiple
flows.

Note that the network state probability distribution,
$\beta_k(\tilde{G})$ in \eqref{eq:sl_net_est_update} or
$\beta_{k-1|k-1}(\tilde{G})$ in \eqref{eq:ge_net_est_innov}, does not
need to converge to a probability distribution describing one topology
realization to yield precise packet predictions
$f_{\vect{\nu}_{^{k}}^{_{k+H-1}}}(\vect{\tilde{\nu}}_{^0}^{_{H-1}})$,
where precise means there is one (or very few similar) high
probability packet delivery sequence(s)
$\vect{\tilde{\nu}}_{^0}^{_{H-1}}$.  Several topology realizations
$\tilde{G}$ may result in the same packet delivery sequence.

Also, note that the GEIHS network estimator performs better when the links
in the network are more bursty.  Long bursts of packet losses from
bursty links result in poor control system performance, which is when
the network estimator would help the most.

%% file: sections/controller_synthesis.tex
\section{FPD Controller}
\label{sec:control_syn}

In this section, we derive the FPD controller using dynamic
programming.  Next, we present two controllers for comparison with the
FPD controller.  These comparative controllers assume particular
statistical models (e.g., i.i.d. Bernoulli) for the packet delivery
sequence pdf which may not describe the actual pdf, while the FPD
controller allows for all packet delivery sequence pdfs.  We derive
the LQG cost of using these controllers. Finally, we present the
computational complexity of the optimal controller.

\subsection{Derivation of the FPD Controller}
We first present the FPD controller and then present its derivation.

\begin{theorem} \label{Thm:fpd}
For a plant with state dynamics given by \eqref{Eq:StateSpace}, the
optimal control policy operating on the information set \eqref{Eq:Ic}
which minimizes the cost function \eqref{Eq:LQGCriterion} results in
an optimal control signal $u_k = -L_k x_k$, where
\begin{equation}
L_k =\big(Q_2 + B^T S_{k+1}\Sst(\nu_k=1,\vect{\nu}_{^0}^{_{k-1}})\D B\big)^{-1}
          B^T S_{k+1}\Sst(\nu_k=1,\vect{\nu}_{^0}^{_{k-1}})\D A \label{Eq:OptL}
\end{equation}
and $S_k : \{0,1\}^k \mapsto \Sbb_+^\ell$ and $s_k : \{0,1\}^k \mapsto \R_+$ 
are the solutions to the cost-to-go at time $k$, given by
\begin{align*}
S_k\Sst(\vect{\nu}_{^{0}}^{_{k-1}})\D ={}& Q_1 + A^T \E\big[S_{k+1}\Sst(\vect{\nu}_{^{0}}^{_{k}})\D|
  \vect{\nu}_{^0}^{_{k-1}}\big]A \\
  {}& - \Pr(\nu_k=1|\vect{\nu}_{^0}^{_{k-1}}) \bigg[A^T S_{k+1}\Sst(\nu_k=1,\vect{\nu}_{^0}^{_{k-1}})\D B
      \big(Q_2 + B^T S_{k+1}\Sst(\nu_k=1,\vect{\nu}_{^0}^{_{k-1}})\D B\big)^{-1} B^T
      S_{k+1}\Sst(\nu_k=1,\vect{\nu}_{^0}^{_{k-1}})\D A\bigg] \\
s_k\Sst(\vect{\nu}_{^{0}}^{_{k-1}})\D ={}& \tr\left\{ \E\big[ S_{k+1}\Sst(\vect{\nu}_{^{0}}^{_{k}})\D|
\vect{\nu}_{^0}^{_{k-1}}\big] R_w\right\} + \E\big[s_{k+1}\Sst(\vect{\nu}_{^{0}}^{_{k}})\D|
\vect{\nu}_{^0}^{_{k-1}}\big] \quad .
\end{align*}
\end{theorem}

\begin{proof}
The classical problem in \AA str\"om \cite{Astrom:2006} is solved by reformulating
the original problem as a recursive minimization of the Bellman
equation derived for every time instant, beginning with $N$. At time
$n$, we have the minimization problem
\begin{align*}
& \quad \min_{_{u_n,\ldots,u_{N-1}}} \; \E[ x_N^T Q_0 x_N + \sum_{i=n}^{N-1} x_i^T Q_1 x_i + \nu_i u_i^T Q_2 u_i ] \\
={}& \E\left[\min_{_{u_n,\ldots,u_{N-1}}} \E[ x_N^T Q_0 x_N + \sum_{i=n}^{N-1} x_i^T Q_1 x_i +
\nu_i u_i^T Q_2 u_i | \Ic^{^\Ce}_n]\right] \\
={}& \E\left[ \quad \; \min_{_{u_n}} \quad \E[ x_n^T Q_1 x_n + \nu_n u_n^T Q_2 u_n + V_{n+1} | \Ic^{^\Ce}_n]\right] \quad ,
\end{align*}
where $V_n$ is the Bellman equation at time $n$. This is given by
\begin{equation*}
V_n = \min_{_{u_n}} \; \E\left[ x_n^T Q_1 x_n + \nu_n u_n^T Q_2 u_n + V_{n+1} | \Ic^{^\Ce}_n\right] \quad .
\end{equation*}

To solve the above nested minimization problem, we assume that the
solution to the functional is of the form $V_n = x_n^T
S_n\Sst(\vect{\nu}_{^0}^{_{n-1}})\D x_n +
s_n\Sst(\vect{\nu}_{^0}^{_{n-1}})\D$, where $S_n$ and $s_n$ are
functions of the past packet deliveries $\vect{\nu}_{^0}^{_{n-1}}$
that return a positive semidefinite matrix and a scalar,
respectively. However, both $S_n$ and $s_n$ are not functions of the
applied control sequence $\vect{u}_{^0}^{_{n-1}}$. We prove this
supposition using induction. The initial condition at time $N$ is
trivially obtained as $V_N = x_N^T Q_0 x_N$, with $S_N = Q_0$ and
$s_N=0$. We now assume that the functional at time $n+1$ has a
solution of the desired form, and attempt to derive this at time
$n$. We have
\begin{align}
V_n ={}& \min_{_{u_n}} \; \E\left[x_n^T Q_1 x_n + \nu_n u_n^T Q_2 u_n +
x_{n+1}^T S_{n+1}\Sst(\vect{\nu}_{^{0}}^{_{n}})\D x_{n+1} + s_{n+1}\Sst(\vect{\nu}_{^{0}}^{_{n}})\D | \Ic^{^\Ce}_n\right] \notag \\
={}& \min_{_{u_n}} \; \E \Big[ x_n^T \big(Q_1 + A^T S_{n+1}\Sst(\vect{\nu}_{^{0}}^{_{n}})\D A \big) x_n +
\nu_n u_n^T \big(Q_2 + B^T S_{n+1}\Sst(\vect{\nu}_{^{0}}^{_{n}})\D B\big) u_n \notag \\
&\quad \quad + \nu_n x_n^T A^T S_{n+1}\Sst(\vect{\nu}_{^{0}}^{_{n}})\D B u_n +
\nu_n u_n^T B^T S_{n+1}\Sst(\vect{\nu}_{^{0}}^{_{n}})\D A x_n + w_n^T S_{n+1}\Sst(\vect{\nu}_{^{0}}^{_{n}})\D w_n + s_{n+1}\Sst(\vect{\nu}_{^{0}}^{_{n}})\D  | \Ic^{^\Ce}_n \Big] \notag \\
={}& \min_{_{u_n}} \; x_n^T \bigg(Q_1 + A^T \E \big[ S_{n+1}\Sst(\vect{\nu}_{^{0}}^{_{n}})\D | \vect{\nu}_{^0}^{_{n-1}} \big] A\bigg) x_n + \tr\bigg\{ \E\big[ S_{n+1}\Sst(\vect{\nu}_{^{0}}^{_{n}})\D|\vect{\nu}_{^0}^{_{n-1}}\big] R_w\bigg\} +
\E[s_{n+1}\Sst(\vect{\nu}_{^{0}}^{_{n}})\D|\vect{\nu}_{^0}^{_{n-1}}] \notag \\
&\quad \quad + \Pr(\nu_n=1|\vect{\nu}_{^0}^{_{n-1}}) \bigg[u_n^T \big(Q_2 + B^T S_{n+1}\Sst(\nu_n=1,\vect{\nu}_{^0}^{_{n-1}})\D
B\big) u_n \notag \\
& \quad \quad + x_n^T A^T S_{n+1}\Sst(\nu_n=1,\vect{\nu}_{^0}^{_{n-1}})\D B u_n + u_n^T B^T
S_{n+1}\Sst(\nu_n=1,\vect{\nu}_{^0}^{_{n-1}})\D A x_n\bigg] \quad \label{Eq:V_n}.
\end{align}
In the last equation above, the expectation of the terms preceded by
$\nu_n$ require the conditional probability
$\Pr(\nu_n=1|\vect{\nu}_{^0}^{_{n-1}})$ and an evaluation of $S_{n+1}$
with $\nu_n=1$. The corresponding terms with $\nu_n=0$ vanish as they
are multiplied by $\nu_n$. The control input at sample time $n$ which
minimizes the above expression is found to be $u_n = -L_n x_n$, where
the optimal control gain $L_n$ is given by \eqref{Eq:OptL}, with $k$
replaced by $n$. Substituting for $u_n$ in the functional $V_n$, we
get a solution to the functional of the desired form, with $S_n$ and
$s_n$ given by
\begin{subequations}
\begin{align}
S_n\Sst(\vect{\nu}_{^{0}}^{_{n-1}})\D ={}& Q_1 + A^T \E\big[S_{n+1}\Sst(\vect{\nu}_{^{0}}^{_{n}})\D|
  \vect{\nu}_{^0}^{_{n-1}}\big]A \notag \\
  {}& - \Pr(\nu_n=1|\vect{\nu}_{^0}^{_{n-1}}) \bigg[A^T S_{n+1}\Sst(\nu_n=1,\vect{\nu}_{^0}^{_{n-1}})\D B
      \big(Q_2 + B^T S_{n+1}\Sst(\nu_n=1,\vect{\nu}_{^0}^{_{n-1}})\D B\big)^{-1} B^T
      S_{n+1}\Sst(\nu_n=1,\vect{\nu}_{^0}^{_{n-1}})\D A\bigg] \label{Eq:S_n}\\
s_n\Sst(\vect{\nu}_{^{0}}^{_{n-1}})\D ={}& \tr\left\{ \E\big[ S_{n+1}\Sst(\vect{\nu}_{^{0}}^{_{n}})\D|
\vect{\nu}_{^0}^{_{n-1}}\big] R_w\right\} + \E\big[s_{n+1}\Sst(\vect{\nu}_{^{0}}^{_{n}})\D|
\vect{\nu}_{^0}^{_{n-1}}\big] \quad .
\end{align}

Notice that $S_n$ and $s_n$ are functions of the variables
$\vect{\nu}_{^0}^{_{n-1}}$. When $n = k$, the current sample time,
these variables are known, and $S_n$ and $s_n$ are not random. But
$S_{n+i}$ and $s_{n+i}$, for values of $i\in\{1,\ldots,N-n-1\}$, are
functions of the variables $\vect{\nu}_{^{0}}^{_{n+i-1}}$, of which
only the variables $\vect{\nu}_{^{n}}^{_{n+i-1}}$ are random variables
since they are unknown to the controller at sample time $n=k$. Since
the value of $S_{n+1}$ is required at sample time $n$, we compute its
conditional expectation as
\begin{equation}
  \E\big[S_{n+1}\Sst(\vect{\nu}_{^{0}}^{_{n}})\D | \vect{\nu}_{^0}^{_{n-1}}\big]
  = \Pr(\nu_n=1|\vect{\nu}_{^0}^{_{n-1}}) S_{n+1}\Sst(\nu_n=1,\vect{\nu}_{^{0}}^{_{n-1}})\D + \Pr(\nu_n=0|\vect{\nu}_{^0}^{_{n-1}}) S_{n+1}\Sst(\nu_n=0,\vect{\nu}_{^{0}}^{_{n-1}})\D \quad . \label{Eq:ExpS}
\end{equation}
The above computation requires an evaluation of
$S_{n+i}\Sst(\vect{\nu}_{^{0}}^{_{n+i-1}})\D$ through a backward
recursion for $i\in\{1,\ldots,N-n-1\}$ for all combinations of
$\vect{\nu}_{^{n+i}}^{_{N-2}}$. More explicitly, the expression at any
time $n+i$, for $i\in\{N-n-1,\ldots,1\}$, is given by
\begin{align*}
\E\big[S_{n+i}\Sst(\vect{\nu}_{^{0}}^{_{n+i-1}})\D|\vect{\nu}_{^0}^{_{n-1}}\big] ={}& Q_1 + A^T
\E\big[S_{n+i+1}\Sst(\vect{\nu}_{^{0}}^{_{n+i}})\D | \vect{\nu}_{^0}^{_{n-1}}\big] A \\
{}& - \sum_{\vect{\tilde{\nu}}_{^0}^{_{i-1}} \in \{0,1\}^i} \Pr\big(\nu_{n+i}=1,\vect{\nu}_{^n}^{_{n+i-1}} =
\vect{\tilde{\nu}}_{^0}^{i-1} | \vect{\nu}_{^0}^{_{n-1}}\big) \\
& \times A^T S_{n+i+1}\Sst(\nu_{n+i}=1,\vect{\nu}_{^n}^{_{n+i-1}} = \vect{\tilde{\nu}}_{^0}^{i-1},
\vect{\nu}_{^0}^{_{n-1}})\D B \\
& \times \big(Q_2 + B^T S_{n+i+1}\Sst(\nu_{n+i}=1,\vect{\nu}_{^n}^{_{n+i-1}} = \vect{\tilde{\nu}}_{^0}^{i-1},
\vect{\nu}_{^0}^{_{n-1}})\D B\big)^{-1} \\
& \times B^T S_{n+i+1}\Sst(\nu_{n+i}=1,\vect{\nu}_{^n}^{_{n+i-1}} = \vect{\tilde{\nu}}_{^0}^{i-1},
\vect{\nu}_{^0}^{_{n-1}})\D A \\
\E\big[s_{n+i}\Sst(\vect{\nu}_{^{0}}^{_{n+i-1}})\D|\vect{\nu}_{^0}^{_{n-1}}\big] &=
\tr\left\{\E\big[S_{n+i+1}\Sst(\vect{\nu}_{^{0}}^{_{n+i}})\D |\vect{\nu}_{^0}^{_{n-1}}\big]R_w\right\} +
\E\big[s_{n+i+1}\Sst(\vect{\nu}_{^{0}}^{_{n+i}})\D|\vect{\nu}_{^0}^{_{n-1}}\big] \quad .
\end{align*}
\end{subequations}

Using the above expressions, we obtain the net cost to be
\begin{equation}
  J = \tr{S_0 R_0} + \sum_{n=0}^{N-1} \tr\left\{
      \E[S_{n+1}\Sst(\vect{\nu}_{^{0}}^{_{n}})\D] R_w\right\} \quad . \label{Eq:LQGCriterion2}
\end{equation}
Notice that the control inputs $u_n$ are only applied to the plant and
do not influence the network or $\vect{\nu}_{^0}^{_{N-1}}$. Thus, the
architecture separates into a network estimator and controller, as
shown in Figure~\ref{fig:system_diagram}.
\end{proof}

\subsection{Comparative controllers}
In this section, we compare the performance of the FPD controller to
two controllers that assume particular statistical models for the
packet delivery sequence pdf, the IID controller and the ON
controller.

\noindent \emph{IID Controller:} The IID controller was described in Schenato et al.
\cite{schenato-procIEEE:2007} and assumes that the packet deliveries
are i.i.d. Bernoulli with packet delivery probability equal to the a 
priori probability of delivering a packet through the
network.\footnote{Using the stationary probability of each link under
the G-E link model to calculate the end-to-end probability of 
delivering a packet through the network.} This is our
first comparative controller, where $u_k = - L_k^\IID x_k$ and the
control gain $L_k^\IID$ is given by
\begin{equation*}
L_k^\IID = \big(Q_2 + B^T S_{k+1}^\IID B\big)^{-1} B^T S_{k+1}^\IID A \quad .
\end{equation*}
Here, $S_{k+1}^\IID$ is the solution to the Riccati equation for the
control problem where the packet deliveries are assumed to be
i.i.d. Bernoulli. The backward recursion is initialized to $S_N^\IID = Q_0$ and
is given by
\begin{equation*}
S_k^\IID = Q_1 + A^T S_{k+1}^\IID A - \Pr(\nu_{k}=1) A^T S_{k+1}^\IID B \big(Q_2 + B^T S_{k+1}^\IID B\big)^{-1} B^T S_{k+1}^\IID A \quad .
\end{equation*}

\noindent \emph{ON Controller:} The ON controller assumes that the
packets are always delivered, or that the network is always
online. This is our second comparative controller, where $u_k = -
L_k^\ON x_k$ and the control gain $L_k^\ON$ is given by
\begin{equation*}
L_k^\ON = \big(Q_2 + B^T S_{k+1}^\ON B\big)^{-1} B^T S_{k+1}^\ON A \quad .
\end{equation*}
Here, $S_{k+1}^\ON$ is the solution to the Riccati equation for the
classical control problem which assumes no packet losses on the
actuation channel. The backward recursion is initialized to $S_N^\ON
= Q_0$ and is given by
\begin{equation*}
S_k^\ON = Q_1 + A^T S_{k+1}^\ON A - A^T S_{k+1}^\ON B \big(Q_2 + B^T S_{k+1}^\ON B\big)^{-1} B^T S_{k+1}^\ON A \quad .
\end{equation*}

\noindent \emph{Comparative Cost:} The FPD controller is the most
general form of a causal, optimal LQG controller that takes into
account the packet delivery sequence pdf.  It does not assume the
packet delivery sequence pdf comes from a particular statistical
model.  Approximating the actual packet delivery sequence pdf with a pdf described
by a particular statistical model, and then computing the optimal control
policy, will result in a suboptimal controller.  However, it may be less
computationally expensive to obtain the control gains for such a
suboptimal controller.  For example, the IID controller and the ON
controller are suboptimal controllers for networks like the one
described in Section~\ref{sec:sys_net_model}, since they presume a
statistical model that is mismatched to the packet delivery sequence
pdf obtained from the network model.

\begin{subequations}
\begin{remark} \label{rem:Cost_sopt}
The average LQG cost of using a controller with control gain
$L_n^{\comp}$ is
\begin{equation} \label{Eq:Cost_sopt}
J = \tr{S_0^\sopt R_0} + \sum_{n=0}^{N-1} \tr\left\{ \E[ S_{n+1}^\sopt \Sst(\vect{\nu}_{^{0}}^{_{n}})\D] R_w\right\} \quad ,
\end{equation}
where
\begin{align}
S_n^\sopt \Sst(\vect{\nu}_{^{0}}^{_{n-1}})\D ={}& Q_1 + A^T \E\big[S_{n+1}^\sopt \Sst(\vect{\nu}_{^{0}}^{_{n}})\D| \vect{\nu}_{^0}^{_{n-1}}\big]A + \Pr(\nu_n=1|\vect{\nu}_{^0}^{_{n-1}}) \notag \\
{}& \times \bigg[ L_n^{\comp^T} \big(Q_2 + B^T S_{n+1}^\sopt \Sst(\nu_n=1,\vect{\nu}_{^0}^{_{n-1}})\D B\big) L_n^{\comp} \notag \\
{}& - A^T S_{n+1}^\sopt \Sst(\nu_n=1,\vect{\nu}_{^0}^{_{n-1}})\D B L_n^{\comp} - L_n^{\comp^T} B^T S_{n+1}^\sopt \Sst(\nu_n=1,\vect{\nu}_{^0}^{_{n-1}})\D A\bigg] \quad , \label{Eq:S_n_corr}
\end{align}
and $\E\big[S_{n+1}^\sopt \Sst(\vect{\nu}_{^{0}}^{_{n}})\D |
\vect{\nu}_{^0}^{_{n-1}}\big]$ is computed in a similar manner to
\eqref{Eq:ExpS}.  The control gain $L_n^{\comp}$ can be the gain of a
comparative controller (e.g., $L_n^\IID$ or $L_n^\ON$) where the
statistical model for the packet delivery sequence is mismatched to
the actual model.  \qed
\end{remark}

This can be seen from the proof of Theorem \ref{Thm:fpd}, if we
substitute for the control input with $u_n^{\sopt}=-L_n^{\comp}x_n$ in
\eqref{Eq:V_n}, instead of minimizing the expression to find the
optimal $u_n$. On simplifying, we get the solution to the cost-to-go
$V_n$ of the form $x_n^T S_n^\sopt \Sst(\vect{\nu}_{^0}^{_{n-1}})\D
x_n + s_n^\sopt \Sst(\vect{\nu}_{^0}^{_{n-1}})\D$, with $S_n^\sopt$
given by \eqref{Eq:S_n_corr} and $s_n^\sopt$ given by
\begin{equation*}
s_n^\sopt \Sst(\vect{\nu}_{^{0}}^{_{n-1}})\D = \tr\left\{ \E\big[ S_{n+1}^\sopt \Sst(\vect{\nu}_{^{0}}^{_{n}})\D|
\vect{\nu}_{^0}^{_{n-1}}\big] R_w\right\} + \E\big[s_{n+1}^\sopt \Sst(\vect{\nu}_{^{0}}^{_{n}})\D|
\vect{\nu}_{^0}^{_{n-1}}\big] \quad .
\end{equation*}
\end{subequations}

\subsection{Algorithm to Compute Optimal Control Gain}
\label{sec:ctrlr_alg}

At sample time $k$, we have $\vect{\nu}_{^0}^{_{k-1}}$. To compute
$L_k$ given in \eqref{Eq:OptL}, we need
$S_{k+1}\Sst(\nu_k=1,\vect{\nu}_{^0}^{_{k-1}})\D$, which can only be
obtained through a backward recursion from $S_N$. This requires
knowledge of $\vect{\nu}_{^k}^{_{N-1}}$, which are unavailable at
sample time $k$. Thus, we must evaluate $\{S_{k+1},\ldots,S_N\}$ for
  every possible sequence of arrivals $\vect{\nu}_{^k}^{_{N-1}}$. This
  algorithm is described below.  \newcounter{counter_alphc}
\begin{enumerate}
\item Initialization: $S_N \Sst(\vect{\nu}_{^k}^{_{N-1}} =
  \vect{\tilde{\nu}}_{^0}^{_{N-k-1}}, \vect{\nu}_{^0}^{_{k-1}})\D = Q_0,$
  $\forall \; \vect{\tilde{\nu}}_{^0}^{_{N-k-1}} \in \{0,1\}^{N-k}$.
\item for $n=N-1:-1:k+1$ \label{Alg:ForLoop}
  \begin{list}{\alph{counter_alphc}) }
  {\usecounter{counter_alphc}}
    \item Using \eqref{Eq:ExpS}, compute
       $\E\big[S_{n+1}\Sst(\vect{\nu}_{^{0}}^{_{n}})\D |
        \vect{\nu}_{^0}^{_{k-1}},\vect{\tilde{\nu}}_{^0}^{_{n-k-1}}\big],$
       $\forall \; \vect{\tilde{\nu}}_{^0}^{_{n-k-1}} \in \{0,1\}^{n-k}$.
    \item Using \eqref{Eq:S_n}, compute
       $S_n \Sst(\vect{\nu}_{^k}^{_{n-1}} = \vect{\tilde{\nu}}_{^0}^{_{n-k-1}},
        \vect{\nu}_{^0}^{_{k-1}})\D,$
       $\forall \; \vect{\tilde{\nu}}_{^0}^{_{n-k-1}} \in \{0,1\}^{n-k}$.
  \end{list}
\item \label{Alg:Sfinal} Compute $L_k$ using $S_{k+1}\Sst(\nu_k=1,\vect{\nu}_{^0}^{_{k-1}})\D$.
\end{enumerate}

For $k=0$, the values $S_0$, $\E[S_1\Sst(\nu_0)\D]$, and the other values
obtained above can be used to evaluate the cost function according to
\eqref{Eq:LQGCriterion2}.

\subsection{Computational Complexity of Optimal Control Gain}
\label{sec:ctrlr_dicuss}

The FPD controller is optimal but computationally expensive, as it
requires an enumeration of all possible packet delivery sequences from
the current sample time until the end of the control horizon to
calculate the optimal control gain \eqref{Eq:OptL} at every sample time
$k$.

\begin{property}
\label{prop:fpd_ctrlr_comp_complex}
The algorithm presented in Section \ref{sec:ctrlr_alg} for computing
the optimal control gain for the FPD controller takes
$O(q^3(N-k)2^{N-k})$ operations at each sample time $k$, where $q =
\max(\ell,m)$ and $\ell$ and $m$ are the dimensions of the state and
control vectors.
\end{property}

\begin{proof}
The computational complexities below assume that multiplying an $\ell
\times m$ matrix with a $m \times n$ matrix takes $O(\ell mn)$, and
that inverting an $\ell \times \ell$ matrix takes $O(\ell^3)$.

For the computation of $L_k$ in \eqref{Eq:OptL}, we need to run the
algorithm presented in Section~\ref{sec:ctrlr_alg}. The steps within
the for-loop (Step~\ref{Alg:ForLoop}) of the algorithm require matrix
multiplications and inversions that take
$O((2\ell^3+6\ell^2m+2\ell^2+2\ell m^2+m^3+m^2)2^{N-k})$ operations, or
$O(q^3 2^{N-k})$ operations if we let $q=\max(\ell,m)$. This must be
repeated $N-k-1$ times in the for-loop, so Step~\ref{Alg:ForLoop}
takes $O(q^3(N-k-1)2^{N-k})$.

Finally, Step~\ref{Alg:Sfinal} takes $O(4\ell^2m+\ell
m^2+m^3+m^2)$ operations for the matrix multiplications and
inversions, which simplifies to $O(q^3)$. Combining these results and
simplifying yields the computational complexity given in
Property~\ref{prop:fpd_ctrlr_comp_complex}.
\end{proof}

For the SIHS network model, once the network state estimates from the SIHS
network estimator converge, the conditional probabilities 
$f_{\vect{\nu}_{^{k}}^{_{k+H-1}}}$ will not change and
the computations can be reused. But, for a
network that evolves over time, like the GEIHS network model, the
computations cannot be reused, and the computational cost remains high.

%% file: sections/examples_simulations.tex
\section{Examples and Simulations}
\label{sec:ex_sim}

Using the system architecture depicted in
Figure~\ref{fig:system_diagram}, we will demonstrate the GEIHS network
estimator on a small mesh network and use the packet delivery
predictions in our FPD controller.  Figure~\ref{fig:ex_mesh_net} depicts
the routing topology and short repeating schedule of the network.
Packets are generated at the source every 409 time
slots,\footnote{Effectively, the packets are generated every $9+4K$
time slots, where $K$ is a very large integer, so we can assume slow
system dynamics with respect to time slots and ignore the delay
introduced by the network.} and the packet delivery deadline is
$t_k'-t_k = 9, \forall k$.  The network estimator assumes all links
have $p^\mathrm{u} = 0.0135$ and $p^\mathrm{d} = 0.0015$.

\begin{figure}
\begin{center}
\begin{tabular}{c}
{\large Routing Topology}\\
\includegraphics[width=5cm]{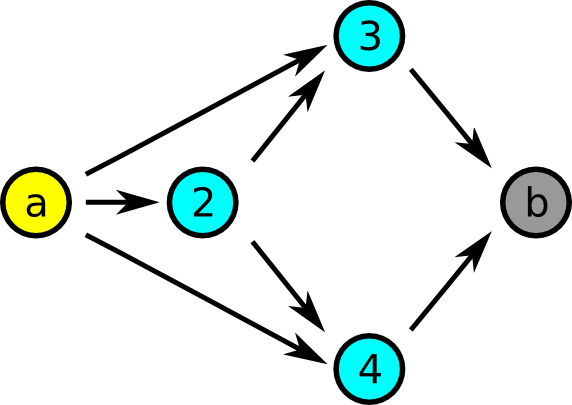}\\
\\
{\large Schedule}\\
\includegraphics[width=10cm]{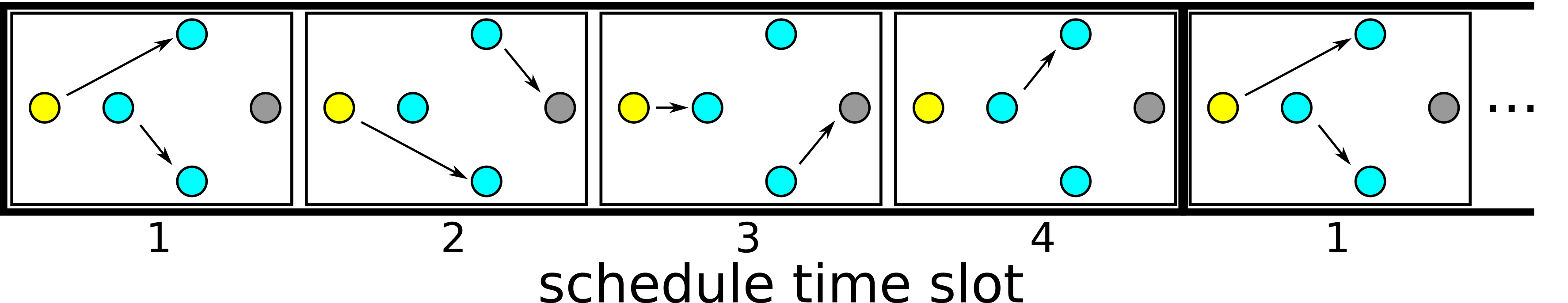}
\end{tabular}
\caption{Example of a simple mesh network for network estimation.}
\label{fig:ex_mesh_net}
\end{center}
\end{figure}

The packet delivery predictions from the network estimator are shown
in Figure~\ref{fig:mesh_net_pkt_pred}.  Although the network estimator
provides
$f_{\vect{\nu}_{^k}^{_{k+H-1}}}(\vect{\tilde{\nu}}_{^0}^{_{H-1}})$, at
each sample time $k$ we plot the average prediction
$\E[\vect{\nu}_{^k}^{_{k+H-1}}]$.  In this example, all the links are
up for $k\in\{1,\ldots,4\}$ and then link $(3,b)$ fails from $k=5$
onwards.  After seeing a packet loss at $k=5$, the network estimator
revises its packet delivery predictions and now thinks there will be a
packet loss at $k=7$.  The average prediction for the packet delivery
at a particular sample time tends toward 1 or 0 as the network
estimator receives more information (in the form of packet deliveries)
about the new state of the network.

\begin{figure}
\begin{center}
\includegraphics[width=10cm]{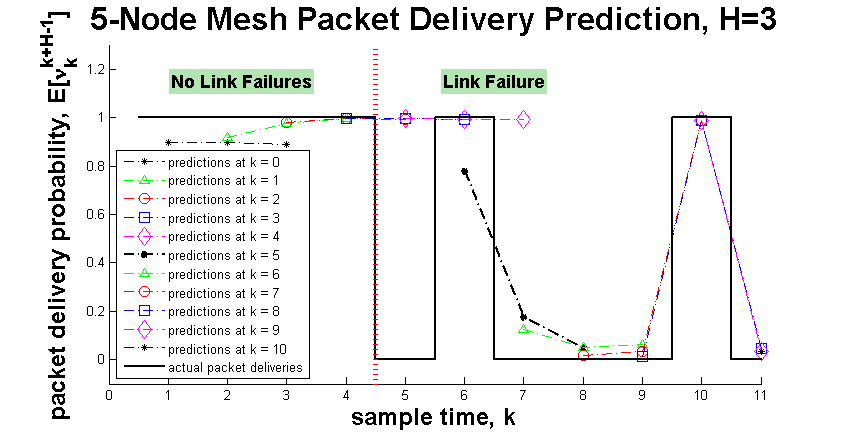}
\caption{Packet delivery predictions when network in
  Figure~\ref{fig:ex_mesh_net} has all links up and then link $(3,b)$ fails.}
\label{fig:mesh_net_pkt_pred}
\end{center}
\end{figure}

The prediction for $k=7$ (packet generated at schedule time slot 3) at
$k=5$ is influenced by the packet delivery at $k=5$ (packet generated
at schedule time slot 1) because hop-by-hop routing allows the packets
to traverse the same links under some realizations of the underlying
routing topology $G$.  Mesh networks with many interleaved paths allow
packets generated at different schedule time slots to provide
information about each others' deliveries, provided the links in the
network have some memory.  As discussed in
Section~\ref{sec:net_est_discuss}, since a packet delivery provides
only one bit of information about the network state, it may take
several packet deliveries to get good predictions after the network
changes.

\begin{figure}
\begin{center}
\includegraphics[width=10cm]{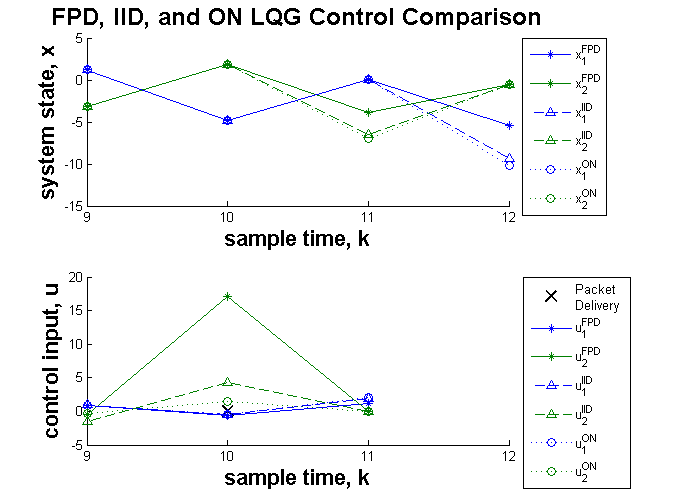}
\caption{Plot of the different control signals and state vectors when
  using the FPD controller, an IID controller, and an ON controller
  (see text for details).}
\label{fig:ControllerComparison}
\end{center}
\end{figure}

Now, consider a linear plant with the following parameters
\begin{small}
\begin{align*}
A = \left[ \begin{array}{cc}
0 & 1.5 \\
1.5 & 0 \end{array} \right],
B = \left[ \begin{array}{cc}
5 & 0 \\
0 & 0.2 \end{array} \right],
R_w &= \left[ \begin{array}{cc}
0.1 & 0 \\
0 & 0.1 \end{array} \right],
R_0 = \left[ \begin{array}{cc}
10 & 0 \\
0 & 10 \end{array} \right] \\
Q_1 = Q_2 = \left[ \begin{array}{cc}
1 & 0 \\
0 & 1 \end{array} \right],
Q_n &= \left[ \begin{array}{cc}
10 & 0 \\
0 & 10 \end{array} \right]
\quad .
\end{align*}
\end{small}
The transfer matrix $A$ flips and expands the components of the state
at every sampling instant. The input matrix $B$ requires the second
component of the control input to be larger in magnitude than the
first component to have the same effect on the respective component of
the state. Also, the final state is weighted more than the other
states in the cost criterion.  We compare the three finite horizon LQG
controllers discussed in Section \ref{sec:control_syn}, namely the FPD
controller, the IID controller, and the ON controller with their costs
\eqref{Eq:LQGCriterion2} and \eqref{Eq:Cost_sopt}.

The controllers are connected to the plant at sample times $k=9,10,11$
through the network example given in Figure~\ref{fig:mesh_net_pkt_pred}.
Figure~\ref{fig:ControllerComparison} shows the control signals computed
by the different controllers and the plant states when the control
signals are applied following the actual packet delivery sequence.
From the predictions at $k=8,9,10$ in
Figure~\ref{fig:mesh_net_pkt_pred}, we see that the FPD controller has
better knowledge of the packet delivery sequence than the other two
controllers.  The FPD controller uses this knowledge to
compute an optimal control signal that outputs a large magnitude for
the second component of $u_{10}$, despite the high cost of this
signal. The IID and ON controllers believe the control packet is
likely to be delivered at $k=11$ and choose, instead, to output a
smaller increase in the first component of $u_{11}$, since this will
have the same effect on the final state if the control packet at
$k=11$ is successfully delivered.

The FPD controller is better than the other controllers at driving the
first component of the state close to zero at the end of the control
horizon, $k=12$. Thus, the packet delivery predictions from the
network estimator help the FPD controller significantly lower its LQG
cost, as shown in Table~\ref{Tb:3C_Cost}. The costs reported here are
obtained from Monte-Carlo simulations of the system, averaged over
10,000 runs, but with the network state set to the one described
above.

\begin{table}%
\begin{center}
\caption{Simulated LQG Costs (10,000 runs) for
         Example Described in Section~\ref{sec:ex_sim}}
\label{Tb:3C_Cost}
\begin{tabular}{| c | c | c |}
\hline
FPD Controller & IID Controller & ON Controller \\ \hline \hline
681.68 & 1,008.2 & 1,158.9 \\ \hline
\end{tabular}
\end{center}
\end{table}

%% file: sections/network_model_selection.tex
\section{Discussion on Network Model Selection}
\label{sec:net_model_select}

The ability of the network estimator to accurately predict packet
deliveries is dependent on the network model.  A natural objection to
the GEIHS network model is that it assumes links are independent and
does not capture the full behavior of a lossy and bursty wireless link
through the G-E link model \cite{willig:2002}.  Why not
use one of the more sophisticated link models mentioned by Willig et al.
\cite{willig:2002}?  Why not use a network model that can capture
correlation between the links in the network?  A good network model
must be rich enough to capture the relevant behavior of the actual
network, but not have too many parameters that are difficult to
obtain.

In our problem setup, the relevant behavior is the packet delivery
sequence of the network.  As mentioned in
Section~\ref{sec:net_est_discuss}, the network state probability
distribution does not need to identify the exact network topology
realization to get precise packet delivery predictions.  In this
regard, the GEIHS network model has too many states ($2^E$ states) and
may be overmodeling the actual network.  However, the more relevant
question is: Does the GEIHS network model yield accurate packet
delivery predictions, predictions that are close to the actual future
packet delivery sequence?  Do the simplifications from assuming link
independence and using a G-E link model result in
inaccurate packet delivery predictions?  These questions need further
investigation, involving real-world experiments.

Our GEIHS network model has as parameters the routing topology $G$,
the schedule $\mathbf{F}^{(T)}$, the G-E link transition probabilities
$\{p_l^\mathrm{u},p_l^\mathrm{d}\}_{l \in \Ec}$, the source $a$, the
sink $b$, the packet generation times $t_k$, and the deadlines $t_k'$.
The most difficult parameters to obtain are the link transition
probabilities, which must be estimated by link estimators running on
the nodes and relayed to the GEIHS network estimator.  Furthermore, on
a real network these parameters will change over time, albeit at a
slower time scale than the link state changes.  The issue of how to
obtain these parameters is not addressed in this paper.

Despite its limitations, the GEIHS network model is a good basis for
comparisons when other network models for our problem setup are
proposed in the future.  It also raises several related research
questions and issues.

Are there classes of routing topologies where packet delivery
statistics are less sensitive to the parameters in our G-E link model
$p_l^\mathrm{u}$ and $p_l^\mathrm{d}$?  How do we build networks
(e.g., select routing topologies and schedules) that are ``robust'' to
link modeling error and provide good packet delivery statistics (e.g.,
low packet loss, low delay) for NCSs?  The latter half of the
question, building networks with good packet delivery statistics, is
partially addressed by other works in the literature like Soldati et
al. \cite{soldati:2010}, which studies the problem of scheduling a network
to optimize reliability given a routing topology and packet delivery
deadline.

Another issue arises when we use a controller that reacts to estimates
of the network's state.  In our problem setup, if the network
estimator gives wrong (inaccurate) packet delivery predictions, the
FPD controller can actually perform \emph{worse} than the ON
controller.  How do we design FPD controllers that are robust to
inaccurate packet delivery predictions?

%% file: sections/conclusions.tex
\section{Conclusions}
\label{sec:conclusions}

This paper proposes two network estimators based on simple network
models to characterize wireless mesh networks for NCSs.  The goal is
to obtain a better abstraction of the network, and interface to the
network, to present to the controller and (future work) network
manager.  To get better performance in a NCS, the network manager
needs to \emph{control and reconfigure} the network to reduce outages
and the controller needs to \emph{react to or compensate for} the
network when there are unavoidable outages.  We studied a specific NCS
architecture where the actuation channel was over a lossy wireless mesh
network and a network estimator provided packet delivery predictions
for a finite horizon, Future-Packet-Delivery-optimized LQG controller.

There are several directions for extending the basic problem setup in
this paper, including those mentioned in 
Sections~\ref{sec:net_est_comp_complex}, \ref{sec:net_est_discuss}, 
and \ref{sec:net_model_select}.  For instance, placing the network
estimator(s) on the actuators in the general system architecture
depicted in Figure~\ref{fig:dist_ctrl_mesh_net} is a more realistic
setup but will introduce a lossy channel between the network
estimator(s) and the controller(s).  Also, one can study the use of
packet delivery predictions in a receding horizon controller rather
than a finite horizon controller.